\documentclass{article}
\usepackage[normalem]{ulem}
\usepackage[utf8]{inputenc}
\usepackage[T1]{fontenc}
\usepackage{lipsum} 
\usepackage[margin=1.5in]{geometry} 
\usepackage{amsmath,amssymb,amsthm}
\usepackage{graphicx}
\usepackage{authblk}
\usepackage{xcolor}
\usepackage{placeins}
\usepackage{url}
\usepackage{caption}
\newtheorem{theorem}{Theorem}
\newtheorem{assumption}{Assumption}
\newtheorem{definition}{Definition}

\title{RiteWeight: Randomized Iterative Trajectory Reweighting for Steady-State Distributions Without Discretization Error}
\author[1]{Sagar Kania}
\author[2]{Robert J. Webber}
\author[3]{Gideon Simpson}
\author[4]{David Aristoff}
\author[1]{Daniel M. Zuckerman\thanks{Corresponding author. Email: zuckermd@ohsu.edu}}

\affil[1]{Department of Biomedical Engineering, Oregon Health and Science University, Portland, OR 97239, USA}
\affil[2]{Department of Mathematics, University of California San Diego, La Jolla, CA 92093, USA}
\affil[3]{Department of Mathematics, Drexel University, Philadelphia, PA 19104, USA}
\affil[4]{Department of Mathematics, Colorado State University, Fort Collins, CO 80523, USA}

\date{\today}

\newcommand{\change}[1]{{\color{black}{#1}}}

\newcommand{\wtnew}{w^\mathrm{new}}

\begin{document}

\maketitle
\section*{Abstract}
A significant challenge in molecular dynamics (MD) simulations
is ensuring that sampled configurations converge to the equilibrium or nonequilibrium stationary distribution of interest. Lack of convergence constrains the estimation of free energies and of rates and mechanisms for molecular transitions. Here, we introduce the “Randomized ITErative trajectory reWeighting'' (RiteWeight) algorithm to estimate a stationary distribution from unconverged simulation data. This method iteratively reweights trajectory segments in a self-consistent way by solving for the stationary distribution of a Markov state model (MSM), updating segment weights, and employing a new random clustering in each iteration.  The repeated clustering mitigates the configuration-space discretization error inherent in existing trajectory reweighting techniques and yields quasi-continuous configuration-space distributions. 
RiteWeight accurately recovers the stationary distribution even without requiring the Markov property at the cluster level.
We present mathematical analysis of the RiteWeight fixed point.
We empirically validate the method using both synthetic MD Trp-cage trajectories, for which the stationary solution is exactly calculable, and standard atomistic MD Trp-cage trajectories, which are extracted from a long reference simulation.  In both test systems, RiteWeight corrects flawed distributions and generates accurate observables for equilibrium and nonequilibrium steady states. The results highlight the value of correcting the underlying trajectory distribution rather than using a standard MSM. 

\section*{Significance}
Molecular dynamics (MD) simulation is a key tool for studying the behavior of proteins and other biomolecules, but despite four decades of hardware and algorithm advances, MD cannot characterize the biomolecular behavior of most systems of interest. Typically, the molecular configurations generated by MD and related methods fail to conform to the equilibrium or nonequilibrium steady state distribution of interest, thereby limiting the accuracy of computed observables such as rate constants. The present report introduces a novel approach for correcting mis-distributed configurations using a Randomized ITErative reWeighting (RiteWeight) strategy.  The approach is validated for protein folding systems under both equilibrium and nonequilibrium conditions.

\section{Introduction}

Despite advances in enhanced equilibrium and path sampling methods \cite{bolhuis2002transition, elber2020milestoning, noe2009constructing, ensing2006metadynamics, kleiman2023adaptive, zuckerman2017weighted, henin2022enhanced}, the study of complex biomolecular systems remains resource-intensive, often surpassing researchers' computational capabilities. Conventional molecular dynamics (MD) cannot produce well-sampled configurational distributions except in special cases \cite{lindorff2011fast}.  A major challenge is estimating the ``steady state'' or stationary distribution of the sampled configurations  \cite{zuckerman2006second, vanGunsteren2006biomolecular, bolhuis2015practical}. The accuracy of the stationary distribution is crucial because it reveals thermodynamic and kinetic properties of the system.
The stationary distribution is needed to calculate free energies as well as identify mechanistic pathways, rate constants, and committor reaction coordinates \cite{russo2021unbiased, voelz2020adaptive, nuske2017markov}.

Recently, methods based on AlphaFold have generated protein structural ensembles \cite{ai_feig2023ensemble,ai_cortes2025alphafold_ensembles,ai_jaakkola024alphafold_ensembles,ai_meiler2023protein_states_ensembles}. However, these methods are heuristic in the sense that they are not designed to produce equilibrium ensembles conforming to the Boltzmann factor.  The algorithm described here, however, can leverage heuristic ensembles as the starting point for conventional MD and subsequent Boltzmann-factor reweighting.  While heuristic starting configurations can be updated with enhanced sampling \cite{tiwary2023alphafold_rave}, all such approaches can be improved by correcting the sample weights.

One principled way to estimate the steady state distribution from sampled trajectory data is to build a ``Markov state model'' (MSM), i.e., a discrete-state transition matrix that leads to an approximation of the stationary distribution \cite{husic2018markov, prinz2011markov}.  However, previous work has noted that MSM estimates of stationarity are biased by the trajectory training data \cite{caflisch2011equilibrium,vitalis2019msm_bias}. 
\change{To partially address this issue, the MSM trajectories can be reweighted based on the matrix's stationary solution \cite{voelz2020adaptive}.}
This procedure --- referred to here as ``single-shot reweighting'' --- can be helpful.
However, it cannot correct the weights of trajectories within the chosen discrete states, nor can it provide a set of weights consistent with a subsequently computed transition matrix~\cite{russo2021unbiased, russo2020iterative}, potentially skewing the estimation of observables.

Unbiased estimation only occurs in the traditional MSM framework when the trajectories in each discrete state are locally consistent with the stationary distribution for the chosen boundary conditions.
For example, source-sink boundary conditions lead to the challenging requirement to sample from the resulting nonequilibrium steady state\cite{russo2021unbiased, voelz2020adaptive, copperman2020accelerated}. 
The need for unbiased equilibrium or nonequilibrium samples creates a chicken and egg problem, which can be addressed by an iterative solution \cite{russo2021unbiased, russo2020iterative} at the cost of generating long trajectories. We note that non-traditional MSMs can be used to obtain unbiased estimates of observables  \cite{nuske2017markov}, but traditional MSMs are biased for both equilibrium and nonequilibrium observables even when significant trajectory data is available \cite{vitalis2019msm_bias,suarez2021markov}.

As an approach for correcting standard MSMs, we here introduce the ``Randomized ITErative trajectory reWeighting'' (RiteWeight) algorithm. RiteWeight reweights trajectories generated without biasing forces into their correct stationary distribution when sufficient data is available.
The method applies to both equilibrium and nonequilibrium steady states.
It naturally uses trajectories of any length, as short as just a single time step, as it does not rely on dynamical relaxation. For example, RiteWeight can employ data generated from standard molecular dynamics or path sampling approaches, as long as no biasing forces have been used.  We emphasize that RiteWeight is distinct from conventional importance sampling, which requires a known and well-sampled initial distribution \cite{swendsen1988single_histogram,zuckerman2010book}.

Figure~\ref{fig:RiteWeight_Algorithm} shows how the RiteWeight algorithm iteratively reweights trajectories using the \change{estimated} stationary measure $\boldsymbol{\pi}$ for a discrete-state transition matrix $\boldsymbol{T}$, which is determined in each iteration by a new random clustering. 
The relative weights of the trajectory segments in each cluster are fixed during a single iteration.
In the equilibrium case, the transition matrix used is akin to a MSM based on weighted trajectories.  In the nonequilibrium case, the matrix is computed from weighted trajectories conforming to source-sink boundary conditions, akin to a ``history augmented'' MSM \cite{suarez2014simultaneous,suarez2016nonMarkov,suarez2021markov,copperman2020accelerated}.
However, RiteWeight uses no history information and thus eliminates the requirement that all segments to be traced back to the source state.
RiteWeight also does not require the Markov property to hold at the cluster level, because discrete states are used only to estimate stationarity and not to propagate dynamics over time.

\begin{figure}[t]
\centering
\includegraphics[width=1.0\linewidth,trim=28 120 50 40,clip]{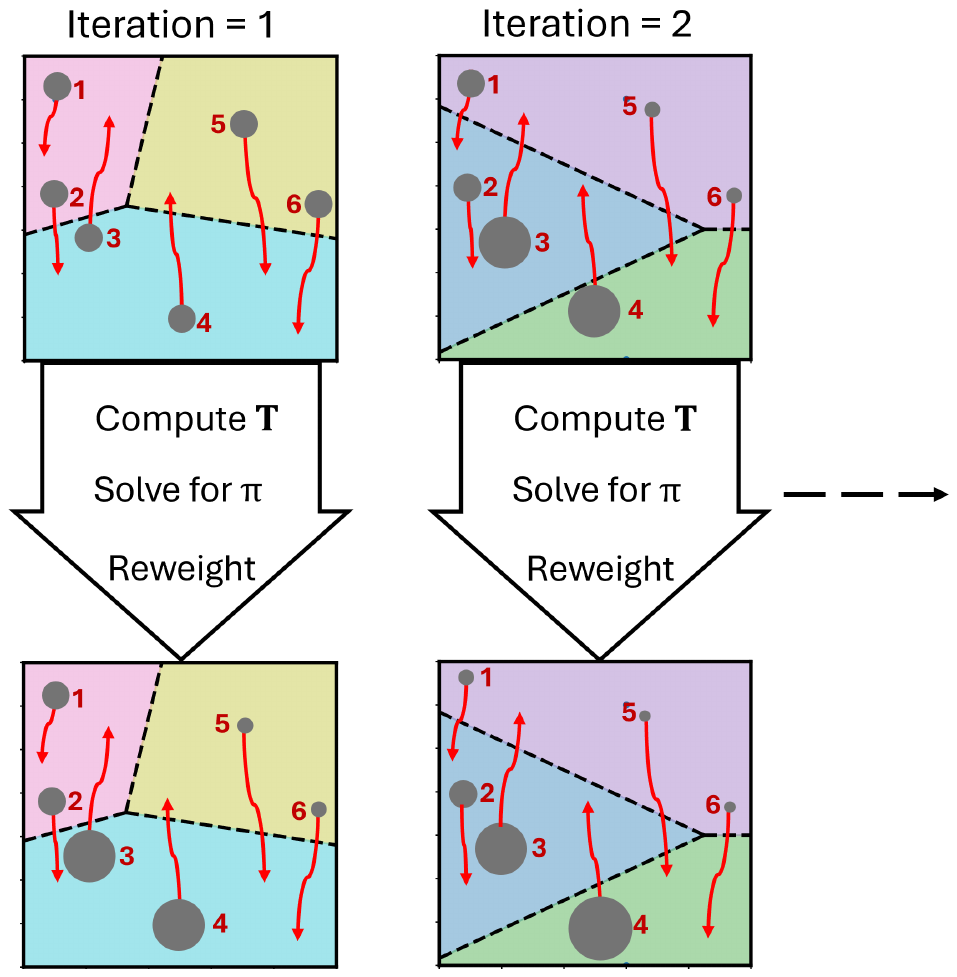}
\caption{The RiteWeight algorithm.  In each iteration, a fixed set of trajectories (red arrows) is organized into clusters (colored regions) based on their starting configurations.  Using the discrete clusters and current weights (circle sizes) of the trajectories, the transition matrix $\boldsymbol{T}$ is computed and solved to yield the stationary measure $\boldsymbol{\pi}$ for the given clusters.  
Each trajectory is then assigned a new weight (filled circles) so that the total cluster weights match with $\boldsymbol{\pi}$ but the relative weights of the trajectories starting within each cluster remain unchanged.  In subsequent iterations, the process is repeated with new cluster boundaries, enabling changes in the relative weights of trajectories formerly in the same cluster, e.g., trajectories 1 and 2 in iteration 2.
}
\label{fig:RiteWeight_Algorithm}
\end{figure}

RiteWeight differs from prior strategies for reweighting sampled trajectory data. For instance, one approach explicitly approximates the configuration-space density function in order to reweight configurations \cite{ytreberg2008blackbox}.  Another method updates the weights using a variational principle for the equilibrium distribution \cite{krivov2021nonparametric}.
RiteWeight, however, does not require density estimation and only uses standard Markov models without additional fitting parameters or assumed functional forms. 
In prior work, our research group has applied iterative reweighting without changing cluster definitions \cite{russo2020iterative}, but this approach does not yield the quasi-continuous distribution which RiteWeight achieves.

The present report evaluates the RiteWeight algorithm through mathematical analysis and empirical testing.
First, we present mathematical analysis identifying the fixed point of the RiteWeight algorithm.
Then, we present numerical tests based on both synthetic molecular dynamics (SynMD) and true MD trajectories of the Trp-cage miniprotein. SynMD consists of trajectories generated from a fine-grained MSM with each state mapped to an atomistic configuration, and it enables comparison to an exactly calculable reference distribution \cite{russo2022simple}. The true MD data is a single 208 $\mu$s trajectory generated by the Shaw group \cite{lindorff2011fast}. We compare the performance of RiteWeight, single-shot reweighting, and traditional MSMs for estimating the stationary distribution from mis-distributed data sets.  Overall, RiteWeight yields better agreement with reference values for all observables considered in both equilibrium and nonequilibrium scenarios.  RiteWeight achieves good performance using extremely short trajectory segments, enabling computation of mechanistic, path-based observables.  Last, RiteWeight results are independent of the number of clusters used to discretize the configuration space.

\FloatBarrier

\section{RiteWeight Algorithm}

The RiteWeight algorithm estimates steady state probability distributions using kinetic information extracted from numerous short MD trajectories or one long trajectory. The algorithm processes sequential pairs of configurations generated by unbiased dynamics and separated by a fixed lag time: $i = (i_1, i_2)$. We refer to these pairs as ``trajectory segments''. Unlike MSMs, the algorithm's ability to estimate stationary probabilities is not constrained by the choice of lag time \cite{prinz2011markov, husic2018markov}. Whereas single-shot reweighting requires within-cluster stationarity for unbiased results \cite{russo2021unbiased, voelz2020adaptive, copperman2020accelerated}, RiteWeight imposes no such constraints. Therefore, transition pair data can be collected right from the start of the trajectories, ensuring that no data is discarded and all available information is leveraged in the analysis.

\change{In the remainder of this report, we use a standardized nomenclature to avoid confusion.  The discrete regions of configuration space employed by RiteWeight will be called ``clusters,'' while the discrete regions used by MSMs will simply be called ``states.''  Collections of states used in the analysis of nonequilibrium fluxes will be called ``macrostates.'' The term ``microstate'' will refer to finest resolution available for describing a system, such as points in configuration space.
Thus, both clusters and states typically consist of multiple microstates, and a macrostate typically consists of multiple states.}

\subsection{The algorithm}

The following steps define the RiteWeight algorithm.
In this algorithm, the user must choose the number of clusters $n$ and identify a featurization that impacts the cluster definitions.

\begin{enumerate}
\item Introduce features of each configuration 
that satisfy rotational and translational invariance. 
For example, C$_\alpha$ pairwise distances are commonly used in conventional MSMs \cite{chodera2014markov_review}.
Define a distance between configurations based on the chosen features.
\item Assign an initial weight $w_i \in (0, 1)$ to each trajectory segment $i$ connecting two consecutive configurations, using the normalization $\sum_i w_i = 1$.
A possible choice is uniform weighting: each $w_i = 1/N$ where $N$ is the number of segments.
\item Randomly select $n \ll N$ unique configurations as cluster centers.
Define clusters by mapping each configuration to the closest center using the chosen distance.
\item Compute a transition matrix ${\boldsymbol T}$ based on the current weights of the trajectory segments and the current definitions of the clusters:
\begin{equation}
\label{rw_transn_matrix}
T_{IJ} = \frac{\sum_{i_1 \in I, \, i_2 \in J} w_i}{\sum_{i_1 \in I} w_i}
\end{equation}
where $i_1 \in I$ means that trajectory segment $i$ begins in cluster $I$, and $i_2 \in J$ means that segment $i$ terminates in cluster $J$. 
\item Calculate the stationary probabilities \(\pi_I\) for the \(n\) clusters using the left leading eigenvector of the matrix ${\boldsymbol T}$, i.e., ${\boldsymbol \pi}^\top \, {\boldsymbol T} = {\boldsymbol \pi}^\top$.

\item Define a new weight for each trajectory segment $i$ based on the current weights and the stationary solution $\boldsymbol{\pi}$ from Step 5. 
For each segment $i$ starting from cluster $I$, the new weight is given as:
\begin{equation}
\label{reweight-eq}
\wtnew_i = 
\frac{\pi_I}{w_I} w_i
\end{equation}
where $w_i$ is the previous weight of the $i^{th}$ trajectory segment, $\wtnew_i$ is the new weight, $\pi_I$ is the current iteration's estimate of the stationary probability for the cluster $I$, and $w_I = \sum_{i_1 \in I} w_i$. \change{Also see the more general update rule in eq.~\eqref{e:alglr}, which introduces an additional learning rate parameter.}
\item Repeat steps 3-6 until a user-defined convergence criterion and average the weights over the final iterations.
\end{enumerate}

Several points are noteworthy. 
In step 2, there are two approaches for assigning the initial weights to trajectory segments.
One option is to assign uniform weights to all segments, which serves as an uninformative prior. Alternatively, if there is prior knowledge regarding the stationary distribution of the system, this information can be used to assign the initial weights in a more informed manner.
Random clusters are defined in step 3, and a new clustering is introduced at every iteration. 
In step 6, the weight of each trajectory in cluster $I$ is multiplied by a fraction $\pi_I/w_I$, where $\pi_I$ is the currently estimated stationary probability and $w_I$ is the previous estimate for the cluster. 

Last, we define convergence in the numerical experiments by running RiteWeight until the weights or mean first passage time estimate exhibits small, steady fluctuations (typically $10^4$--$10^6$ iterations).
Then, we average the weights over the final 1,000 or 10,000 iterations to further reduce noise.
See supplementary Figs.~\ref{fig:synmd-converge-kl}, \ref{fig:RiteWeight_Eq_Conv}, and \ref{fig:MFPT_Conv} for examples of RiteWeight convergence analysis.

\subsection{Analysis of RiteWeight fixed point}

The mathematical analysis of RiteWeight illustrates the strengths and limitations of the algorithm.
We analyze the fixed point in detail in the appendix, and we summarize the findings here.

For simplicity and clarity, our analysis assumes that the configurations all belong to a discrete space of ``microstates'' indexed by $\alpha$ or $\beta$.
We expect that RiteWeight in continuous space behaves similarly to RiteWeight in discrete space at a level of resolution automatically determined by the algorithm and its stopping criterion.

In the discrete setting, RiteWeight defines each update using coarse random clusters containing multiple microstates.
However, independent of the number of clusters, the fixed point of RiteWeight is determined by the stationary measure of the microstate transition matrix with elements
\begin{equation}
\label{eq:microstate_transition}
    P_{\alpha\beta} = \frac{\sum_{i_1 \in \alpha, \, i_2 \in \beta} w_i^{\rm init}}{\sum_{i_1 \in \alpha} w_i^{\rm init}}.
\end{equation}
Here, $w_i^{\rm init}$ is the \emph{initial} weight assigned to segment $i$.
The appendix proves the fixed point of RiteWeight equals the stationary vector of $\boldsymbol{P}$, assuming the stationary vector is unique and the clusters are chosen according to a random hyperplane tesselation.

The fixed point analysis suggests that RiteWeight will converge to the true stationary distribution given sufficiently dense and unbiased local sampling.
\change{In the limit of high transition counts, the trajectory segments will be described by a fine-grained transition matrix $\boldsymbol{P}$ that approaches the true transition operator.}
Thus, RiteWeight will find the correct distribution regardless of the number of clusters or the lag time.
However, in more realistic settings, the accuracy of RiteWeight is constrained by the data analyzed,
pointing to a valuable role for the adaptive sampling methods that will be mentioned in the discussion section.

Given the description of the RiteWeight fixed point, it is natural to ask whether the RiteWeight stationary distribution can be derived directly from a fine-grained MSM transition matrix.
In practice, this is not possible because the resolution of the microstates captured by RiteWeight cannot be easily determined for an arbitrary data set. 
The results section will show the limitations of simply using a fine-grained Markov state model without RiteWeight.

\section{Test systems} \label{sec:systems}

Here we introduce the synthetic MD and atomistic MD systems used in our empirical RiteWeight experiments.

\subsection{Synthetic MD for Trp-cage}

The first test system for RiteWeight is synthetic MD (SynMD).
SynMD offers some of the complexity of atomistic proteins, but it has computational and conceptual advantages.
Using this system, long trajectories can be efficiently generated, and stationary properties can be calculated exactly as a reference \cite{russo2022simple}.

\emph{Model.} The SynMD model for the Trp-cage miniprotein \cite{russo2022simple} is derived from an MSM trained on a 208 $\mu$s atomistic MD trajectory \cite{lindorff2011fast}.
The MSM uses 10,500 states, which is a finer discretization than a typical MSM, and each state is mapped to a specific atomistic configuration. 
\change{These 10,500 configurations are then used as microstates for SynMD}.
Kinetic Monte Carlo yields a ``synthetic'' MD trajectory that closely matches the statistical behavior of the original atomistic MD \cite{russo2022simple}.
The MSM lag time and hence the interval between synMD configurations is 1 ns, and the trajectories can be processed using standard MSMs or RiteWeight.
Because the synMD model is governed by an explicitly known MSM transition matrix, we can exactly calculate the equilibrium distribution \change{for microstates} using the leading left eigenvector of the matrix.

\emph{Trajectory preparation.} 
We prepared a data set of 5 ns SynMD trajectories using a sampling measure far from equilibrium.
First we applied tICA with a 5 ns lag time \cite{suarez2021markov} to the set of minimal residue-residue distances, which are defined by the closest distances between the heavy atoms of two residues separated in sequence by at least two neighboring residues.
Then we assigned indices to the 10,500 microstates based on the slowest time-lagged independent component (TIC$_1$) \cite{russo2022simple}.
We partitioned the SynMD indices into 500 bins and selected 20 initial microstates randomly with replacement from each bin.
For each selected microstate, we initiated a short trajectory of 5 ns, i.e., 5 steps. 
Hence the unprocessed trajectory data consists of 10,000 trajectory segments, each 5 steps long, distributed roughly evenly along the tIC$_1$ coordinate.

\emph{RiteWeight analysis.} We applied RiteWeight in a manner blind to the discrete nature of synMD, using only the mapped atomistic configurations.
We defined features using the slowest 10 tICs, and we rescaled the tICs based on commute distances.
Then we performed RiteWeight using either $n = 10$ or $n = 1000$ clusters based on the rescaled tICs.

\subsection{Atomistic MD for Trp-cage}

The second test system for RiteWeight is a 208 $\mu$s simulation of Trp-cage in explicit solvent \cite{lindorff2011fast}.
We used this atomistic MD trajectory data to investigate equilibrium and nonequilibrium steady states, as well as path-based observables including the mean first-passage time and net fluxes.

\subsubsection{Equilibrium distribution}
\emph{Trajectory preparation.}
Again we prepared a data set of short trajectories using a sampling measure far from equilibrium.
To that end, we applied tICA with a 10 ns lag time \cite{suarez2021markov} to the set of minimal residue-residue distances.
Then we projected all configurations from the 208 $\mu$s trajectory onto tIC$_1$ and grouped them into 100 uniformly spaced bins.
We subsampled 10,000 configurations within each bin;
however, if fewer than 10,000 configurations were available in a bin, we selected all the configurations. 
We used these configurations as the starting points for segments consisting of two snapshots from the long trajectory separated by a 10 ns lag time.

\emph{RiteWeight and MSM analysis.} 
We retained a sufficient number of tICs to explain 95\% of the total variance and rescaled the tICs based on commute distances. For RiteWeight, we performed random clustering ($n = 10$) using the rescaled tICs as features. For the MSM analysis, we performed $k$-means clustering based on $n = 100$--$50{,}000$ clusters. 
After clustering, we used the PyEMMA software package \cite{noe2015pyemma} for all MSM analysis.

\subsubsection{Nonequilibrium analysis}\label{subsec:nonequilibrium_analysis}

\emph{Trajectory preparation.} We harnessed the complete equilibrium-like 208 $\mu$s MD trajectory to generate trajectories for nonequilibrium steady state analysis.
No subsampling of the trajectory data was performed.
Instead, we extracted all possible two-step segments based on a lag time of $\tau=$ 100, 10, 1, or 0.2 ns (the shortest lag time in the original MD data).

\change{
\emph{MSM analysis and definition of ``folded'' and ``unfolded'' states.} 
Previous work has identified optimal MSM hyperparameters for the 208 $\mu$s MD data set \cite{husic2018markov,mey2024msm_optimize}.
We replicated the optimal parameter choices with the PyEMMA software package \cite{noe2015pyemma}, including features based on minimal residue-residue distances, dimensionality reduction using tICA with 100 tICs at a lag time of 10 ns, rescaled tICs based on commute distances, and clustering using $k$-means with $k=50$.
Next, following the prior work \cite{suarez2021markov}, we applied the fuzzy spectral clustering method PCCA++ with a lag time of 100 ns to identify two extreme states representing ``folded'' and ``unfolded'' structures.
PCCA++ also assigned folded and unfolded membership scores to each of the $48$ remaining MSM states.

\emph{RiteWeight analysis.}
In each RiteWeight iteration, we performed clustering using $n = 10$ random clusters based on the same tICA coordinates used for MSM analysis. 
We only clustered configurations in the intermediate region, excluding the previously identified MSM folded and unfolded states.
We then constructed a $12 \times 12$ transition matrix that included the intermediate clusters along with the folded and unfolded states. To enforce the non-equilibrium source-sink condition, we set the transition probability from the source to the sink equal to one and set the transition probability from the source to all other clusters equal to zero.

\emph{Flux analysis.}
Followed prior work \cite{suarez2021markov}, we identified $7$ ``macrostates'' for flux analysis with both MSMs and RiteWeight.
First we selected the $5$ intermediate MSM states whose PCCA++ membership scores were closest to 50\% folded / 50\% unfolded.
Then we grouped the remaining $45$ MSM states into a folded and an unfolded macrostate based on their membership scores.}

\section{Results}

Here we present the results of applying RiteWeight to the SynMD and atomistic MD systems introduced in Sec.~\ref{sec:systems}.

\subsection{Synthetic MD for Trp-cage}

Fig.\ \ref{fig:synmd-equil-compare} shows that
RiteWeight recovers the true equilibrium distribution for the SynMD Trp-cage system starting from trajectory data that is strongly mis-distributed.
Although the algorithm employs discrete MSM-like stationary solutions at each iteration,
it learns a quasi-continuous distribution.

RiteWeight's accuracy is equally high using 10 or 1,000 clusters, underscoring the algorithm's robustness.
However, RiteWeight's convergence speed \emph{does} depend on the number of clusters, with more clusters yielding faster convergence.
See the convergence analysis in Supplementary Fig.~\ref{fig:synmd-converge-kl}.
Following this analysis, we used $10^4$ iterations for RiteWeight with 1000 clusters and $10^6$ iterations for RiteWeight with 10 clusters.
In both cases, the trajectory weights reported here are averages over the final 1,000 iterations.

\begin{figure}[t]
\centering
\includegraphics[width=1.0\linewidth, trim=45 30 50 55,clip]{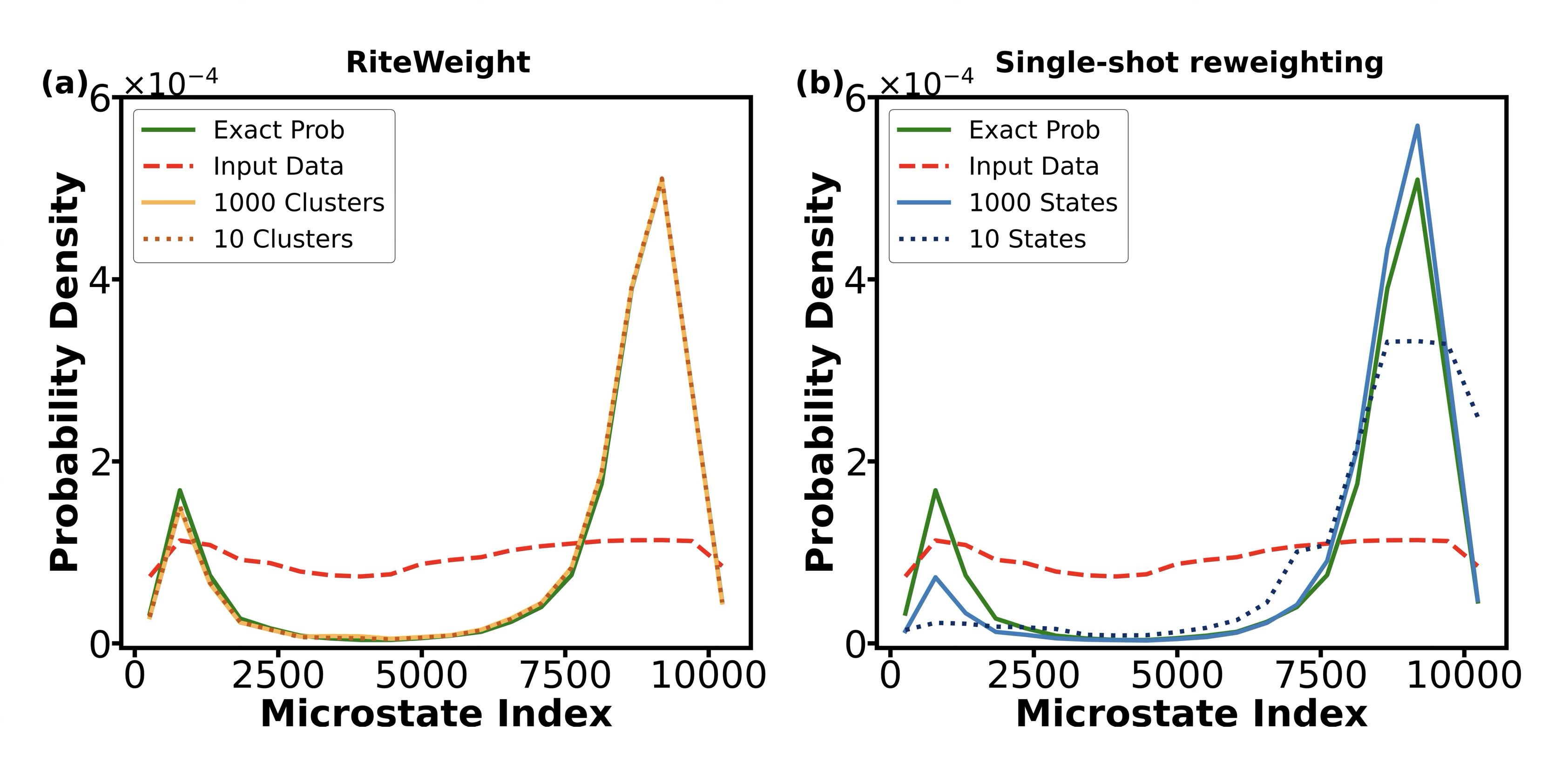}
\caption{Estimates of the equilibrium distribution for the SynMD Trp-cage system. (a) RiteWeight recovers the true equilibrium distribution (green solid line) starting from a \change{highly nonequilibrium initial distribution} (red dashed), using either 10 (dark orange dashed) or 1,000 clusters (orange solid).
Here, 
the lag time is 1 ns.
(b) Single-shot reweighting deviates from the true distribution using 10 (dark blue dashed line) or 1,000 (light blue solid) clusters. Note that the \change{micro}state index is ordered according to tIC$_1$.
}
\label{fig:synmd-equil-compare}
\end{figure}

In contrast to the converged RiteWeight estimates, the single-shot reweighting that occurs during the first RiteWeight iteration exhibits discrepancies due to discretization error. 
Even with 1,000 clusters, single-shot reweighting does not accurately predict the true stationary probabilities of the Trp-cage SynMD microstates.
The discrepancy here can be attributed to the lack of local equilibrium within the \change{discrete clusters}:
local equilibrium is necessary for the success of single-shot reweighting but not for RiteWeight. 

\FloatBarrier
\subsection{Atomistic MD for Trp-cage: Equilibrium and nonequilibrium}

RiteWeight also recovers the correct stationary distribution for atomistic Trp-Cage both in equilibrium and nonequilibrium settings.

\change{Fig.~\ref{fig:equil-MD} shows estimates of the equilibrium distribution along the slowest tIC (tIC$_1$), calculated using RiteWeight and an MSM. Supplementary Fig.~\ref{fig:2D_landscape} similarly shows a 2D projection of these estimates onto the two slowest tICs.}
The reference values come from a direct computation using the full 208 $\mu$s Trp-Cage trajectory.
RiteWeight matches closely with the reference despite starting from a distribution far from equilibrium.
Here, the RiteWeight distribution is an average over the final 10,000 iterations in a run of 100,000 iterations. See the convergence analysis in Supplementary Fig.~\ref{fig:RiteWeight_Eq_Conv}.

\begin{figure}[t]
\centering
\includegraphics[width=0.8\linewidth]{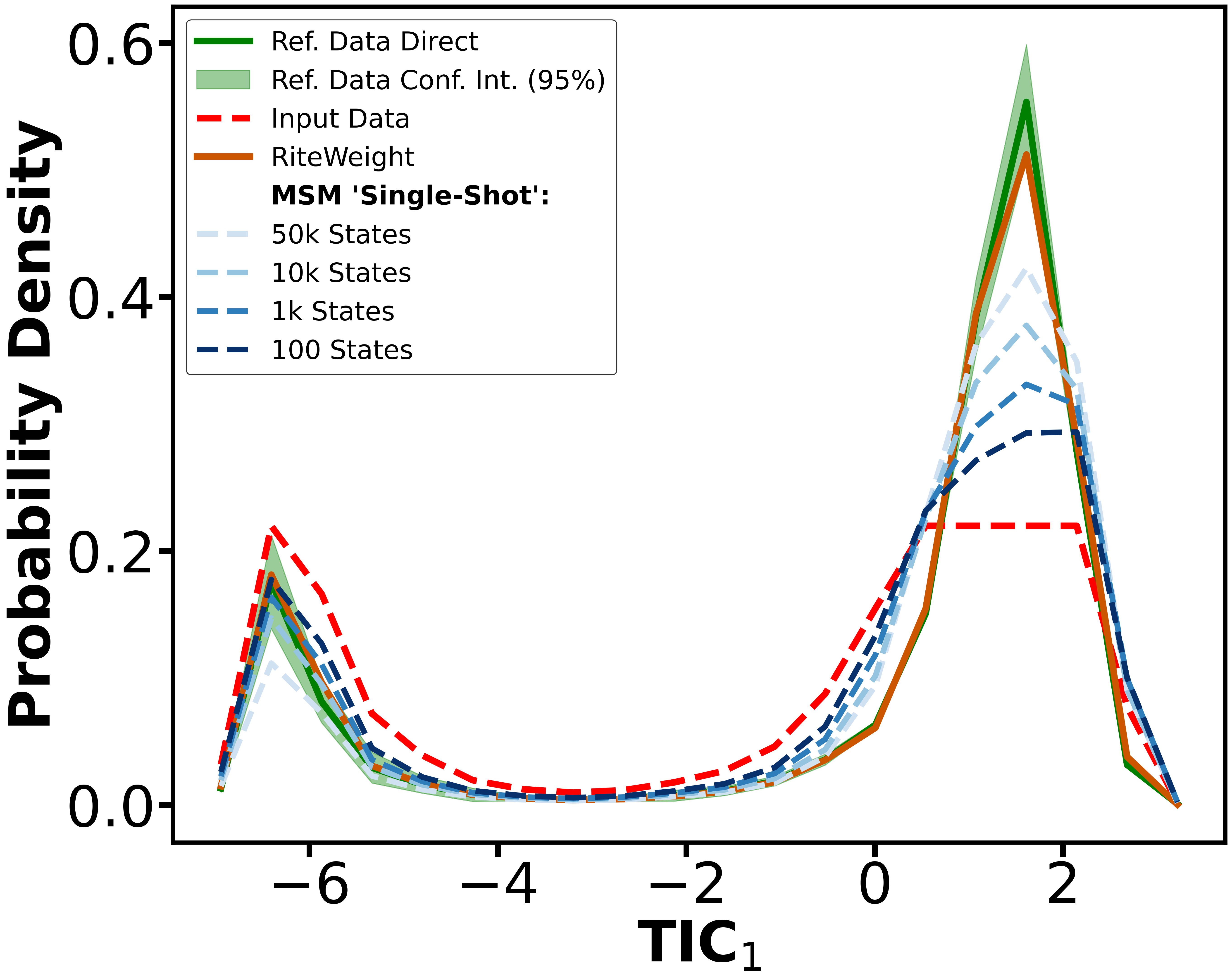}
\caption{Estimates of the equilibrium distribution for the atomistic Trp-cage system. 
Reference data is from a 208 $\mu$s MD trajectory (green with shaded error bars).
Starting from mis-distributed input data (red dashed),
RiteWeight recovers the true distribution (dark orange).
Also shown are MSM ``single shot'' estimates based on 100--50,000 states (blue dashed lines). 
Both RiteWeight and the MSM employ a lag time of 10 ns, and RiteWeight uses $n = 10$ clusters.
}
\label{fig:equil-MD}
\end{figure}

By contrast, MSM-based reweighting fails to recover the equilibrium distribution, even when using the same features as RiteWeight with as many as 50,000 clusters.
The inability to recover the correct stationarity distribution is due to an unsatisfied assumption that weights are locally in equilibrium within each \change{MSM state}.
The poor performance persists even when the MSM lag time is set to 100 ns; see Supplementary Fig.~\ref{fig:Eq_MSM_TR_100ns}. 

\change{The equilibrium stationary distribution is characterized by a detailed balance condition \cite{zuckerman2010book}.
However, RiteWeight does not enforce detailed balance directly, as do the MSM tools used for data analysis in this study \cite{noe2015pyemma}.
Nevertheless, Supplementary Fig.~\ref{fig:detail_balance} shows that RiteWeight recovers detailed balance to a similar or greater extent than a MSM.}

Fig.~\ref{fig:rw_ness} shows that RiteWeight is also successful in the nonequilibrium setting.
Here, the reference stationary measure comes from modifying the 208 $\mu$s Trp-cage trajectory by removing configurations that were more recently in the folded than the unfolded state.
Again RiteWeight corrects the imbalances in the starting distribution and delivers a faithful copy of the stationary distribution.
Here, the RiteWeight distribution is an average over the final 10,000 iterations in a run of 400,000 iterations, based on the mean first passage time convergence analysis in Supplementary Fig.~\ref{fig:MFPT_Conv}.

\begin{figure}[t]
\centering
\includegraphics[width=0.8\linewidth]{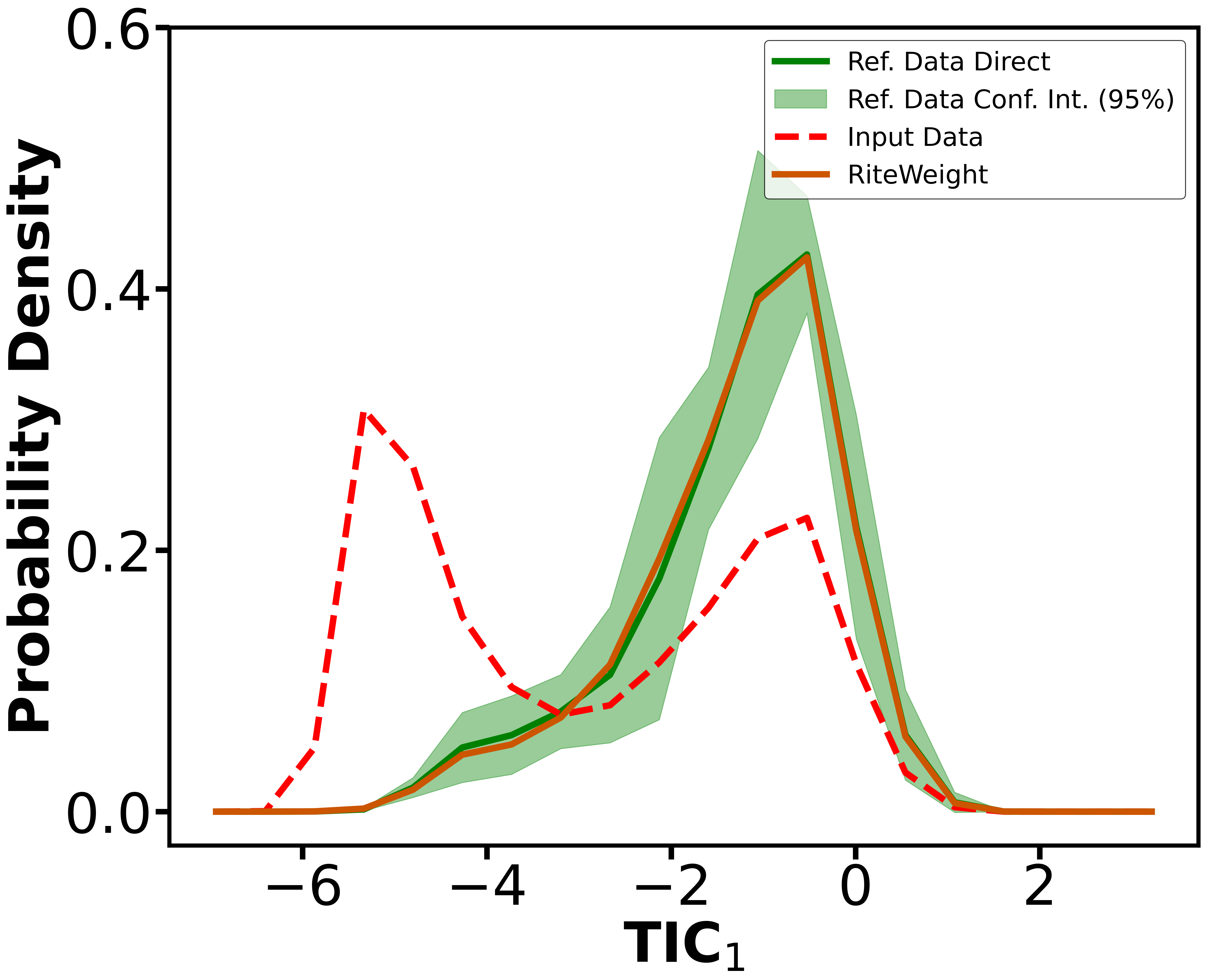}
\caption{Estimates of the nonequilibrium distribution for the atomistic Trp-cage system.
The plot shows the probability density for the intermediate region, after excluding the folded and unfolded states.
The initial distribution (red) represents all MD samples outside the folded and unfolded states, while the MD reference distribution (green) is derived from trajectory segments more recently in the unfolded than the folded state.
RiteWeight (dark orange) closely follows the reference distribution. 
}
\label{fig:rw_ness}
\end{figure}

\subsection{Atomistic MD for Trp-cage: Mean first-passage time}

Next we used RiteWeight to calculate the mean first-passage time (MFPT) for atomistic trp-Cage folding.
We ran RiteWeight using source-sink boundary conditions with the unfolded region as the source and the folded region as the sink.
We then evaluated the MFPT using the reciprocal flux relation \cite{hill2005free,bhatt2010steady}
\begin{equation}
\label{eq:mfpt}
\operatorname{MFPT}
= \Biggl(\frac{1}{\tau} \sum_{i_1 \notin \text{folded}, \, i_2 \in \text{folded}} w_i \Biggr)^{-1},
\end{equation}
where $\tau$ is the lag time and $w_i$ is the weight assigned to each trajectory segment $i$. For comparison, we also performed an MSM analysis using the first-step relation
\begin{equation}
\operatorname{MFPT}
= M_{\text{unfolded}},
\quad \text{where} \quad
M_I = \begin{cases}
    0, & I = \text{folded}, \\
    \tau + \sum_J T_{IJ} M_J, & I \neq \text{folded},
\end{cases}
\end{equation}
which is implemented in the PyEMMA software \cite{noe2015pyemma}.

Fig.~\ref{fig:mpft} shows that RiteWeight recovers the reference MFPT values, even with a short lag time of $\leq 1$ ns.
We note that the reference MFPT from the 208 $\mu$s trajectory slightly increases with lag time, because some of the first entries to the target folded state are missed at longer lag times.
Consistent with previous work \cite{suarez2021markov}, the MSM recovers the correct MFPT but only at sufficiently long lag times $\geq 100$ ns.
At short lag times $\leq 1$ ns, the MSM underestimates the MFPT by an order of magnitude or more.

\begin{figure}[t]
\centering
\includegraphics[width=0.8\linewidth]{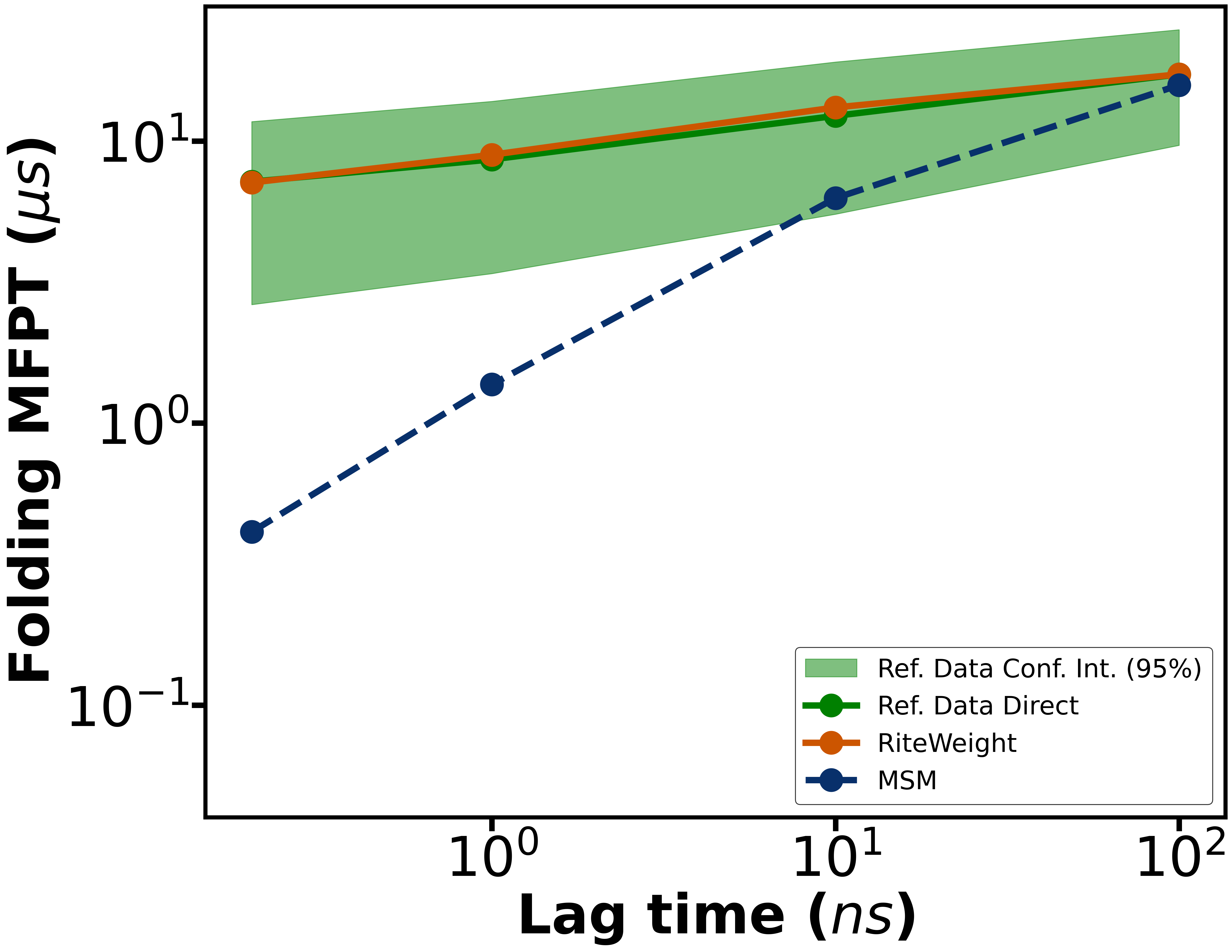}
\caption{Estimates of the mean first passage time for the atomistic Trp-cage system to fold.
Reference MD values (dark green) are based on 20 or fewer events with associated uncertainty (green shaded region). The reference values are compared to RiteWeight estimates based on $10$ clusters (dark orange) and MSM estimates based on $50$ \change{states} (blue).}
\label{fig:mpft}
\end{figure}

RiteWeight is able to reproduce the MFPT because it uses a self-consistent iteration that locally corrects the weights within each cluster.
The initial weights are close to equilibrium, but the nonequilibrium distribution from the source-sink boundary conditions is ``tilted'' with respect to equilibrium \cite{vandenEijnden2009tilting}.
We conjecture that updating the weights is especially important at short lag times because the trajectories determining the MFPT estimate are the peripheries of the clusters where the tilting is greatest.

\subsection{Atomistic MD for Trp-cage: Net fluxes}

A major goal of molecular simulation is understanding the ``mechanisms'' of a functional transition, here defined as the temporal sequences of the molecule.
We analyzed the mechanisms for atomistic Trp-cage folding using net fluxes,
\begin{equation}
\mbox{Net flux}(I,J) = \frac{1}{\tau} \Biggl(\sum_{i_1 \in I, \, i_2 \in J} w_i - \change{\sum_{i_1 \in J, \, i_2 \in I}} w_i \Biggr),
\label{netflux}
\end{equation}
where the macrostate indices $I$ and $J$ were ordered to produce a positive net flux in the reference MD data. \change{We analyzed fluxes involving a folded macrostate, an unfolded macrostate, and five intermediate macrostates, as described in Sec.~\ref{subsec:nonequilibrium_analysis} (\emph{Flux analysis}).
Supplementary Fig.~\ref{fig:flux_network} shows the intermediate macrostates projected onto the two slowest tICA components, and it illustrates the sequences of transitions from the unfolded to the folded macrostate.}

For comparison, we also performed an MSM analysis using the formula for the net fluxes involving the backward committor \cite{suarez2021markov}:
\begin{equation}
    \mbox{Net flux}(I,J) = \frac{1}{\tau} \Bigl(q_I^{(-)} \pi_I^{(\mathrm{eq})} T_{IJ} - q_J^{(-)} \pi_J^{(\mathrm{eq})} T_{JI} \Bigr).
\end{equation}
Here $\boldsymbol{\pi}^{(\mathrm{eq})}$ is the equilibrium steady state, which comes from solving $(\boldsymbol{\pi}^{(\mathrm{eq})})^\top \boldsymbol{T} = (\boldsymbol{\pi}^{(\mathrm{eq})})^\top$, and
$\boldsymbol{q}^{(-)}$ is the backward committor vector computed by pyEmma \cite{noe2015pyemma}, which gives the probability that a Markov chain with transition matrix $\boldsymbol{T}$ was more recently in the unfolded state than the folded state.

\begin{figure}[t]
\centering
\includegraphics[width=1.0\linewidth]{figures/NetFlux.png}
\caption{Estimates of the mechanistic fluxes for atomistic Trp-cage folding.
Estimates from RiteWeight (dark orange) and MSMs (blue)
are compared to reference values computed from the 208 $\mu$s MD trajectory.
The black dashed line represents equality of predicted and reference values.
The panels show lag times $\tau=$ 0.2, 1, 10, and 100 ns.}
\label{fig:netflux}
\end{figure}

Fig.~\ref{fig:netflux} shows that the net fluxes from RiteWeight closely match the MD reference for each lag time $\tau=$ 0.2, 1, 10, or 100 ns.
Thus, RiteWeight precisely describes the transition processes leading to folding, even processes with a lag time $\tau =0.2$ ns which is the shortest lag time available in the original MD data.
In contrast, the net fluxes from MSMs exhibit large discrepancies, echoing a previously identified deficiency of the MSM approach \cite{suarez2021markov}.  
At each lag time, the MSMs exhibit some net fluxes in the opposite direction of MD, implying a different temporal sequences of events.
The MSM fluxes are most quantitatively accurate at the 100 ns lag time, which is 1-2 orders of magnitude longer than the transition-path time over which folding events typically take place \cite{suarez2021markov}.
Supplementary Figs.~\ref{fig:net_flux_CI_short_lags} and \ref{fig:net_flux_CI_long_lags} report confidence intervals for the MD reference fluxes, revealing tht a range of MSM estimates fall outside the 95\% confidence intervals at each lag time $\tau \leq$ 10 ns.

\FloatBarrier

\section{Discussion and conclusions}

We have introduced RiteWeight, an algorithm for correcting distributions that deviate from equilibrium or nonequilibrium steady states. Although it is built on a framework of Markov state models (MSMs) that discretize configuration space, RiteWeight iteratively achieves self-consistency among all possible discretizations.
This new approach overcomes a key limitation of MSMs by correcting the distribution within each discrete region to conform to the steady state.
It accurately evaluates the steady state without requiring the Markov property at the cluster level.
When the internal distributions match the steady state, a wider array of observables --- such as mean first passage times and net fluxes --- can be calculated with high accuracy even at short lag times \cite{suarez2021markov}. The robustness of RiteWeight is illustrated by the insensitivity to the number of clusters used during the iteration process (Fig.\ \ref{fig:synmd-equil-compare}).

\change{Although RiteWeight employs coarse random clustering, over multiple iterations it adjusts the relative weights of trajectories at a fine resolution in feature space. Nevertheless, extremely close trajectories are unlikely to be separated into different clusters, so their relative weights are unlikely to change. To ensure the correction of weights, the feature space should be defined to differentiate configurations that are separated by a significant energy barrier. For example, in the Trp-Cage system, we constructed a suitable feature space using time-lagged independent components (tICs) derived from the dynamical behavior of the observed trajectories.}

RiteWeight produces models analogous to ``history augmented'' MSMs \cite{dinner2009separating,vandenEijnden2009tilting,suarez2014simultaneous,suarez2021markov} but with a key difference: no history information is used to analyze the data.  
For this reason, RiteWeight is capable of processing a wide variety of dynamics data, even data from adaptive sampling \cite{pande2010adaptive_msm,clementi2018adaptive}.
The raw data can be generated from multiple trajectories of arbitrary lengths initiated from arbitrary starting points.

The main requirement of the current implementation is that trajectory dynamics should be unbiased, i.e., outcomes from a given starting point should reflect the distribution that would occur in ``ordinary'' simulation without any unphysical forces or imposed stopping criteria.  Thus, for example, weighted ensemble (WE) data could be used but only for lag times up to the WE resampling time; for longer lag times, a statistical continuation analysis would be required \cite{dickson2025wemergebias}. However, it should also be possible to use Girsanov reweighting of biased trajectories, a strategy previously applied for MSM construction \cite{keller2017girsanov}.

Building on previous work \cite{suarez2021markov}, we have probed the shortcomings of MSMs for identifying mechanisms of folding in the atomistic Trp-cage system (Fig.~\ref{fig:netflux}).
Surprisingly, the MSMs built from MD reference data continue to show discrepancies from reference values even at lag times of $\tau \geq $ 100 ns.
At sufficiently long lag times, the system must become fully Markovian, but evidently at 100 ns, the net fluxes are still influenced by the initial distributions within the \change{MSM states}.
The sensitivity of the net flux estimates contrasts with the relative robustness of MSM mean first passage time estimates that match the reference value at a lag time of 100 ns (Fig.~\ref{fig:mpft}).

There remain important tasks ahead to realize the full potential of RiteWeight.  Most notably, the data analyzed here consists of large data sets with dense sampling in visited regions of configuration space.   In sparser data sets, the algorithm may need to rely on other regularization and configuration-space smoothing \cite{otten2026smooth-rw}.
In addition, we can modify the RiteWeight update formula~\eqref{reweight-eq}
to read
\begin{equation}
\label{e:alglr}
\wtnew_i = 
(1 - r) \, w_i + r \, 
\frac{\pi_I}{w_I} w_i.
\end{equation}
Here, \(r\) is a learning rate hyperparameter within the interval $(0,1]$. 
Choosing a smaller \(r\) value may mitigate the impact of noise without changing the RiteWeight fixed point (Theorem 1 in the appendix).
Another valuable strategy that may improve the estimation of observables in key configuration-space regions is adaptive sampling \cite{pande2010adaptive_msm,clementi2018adaptive}.
Finally, because RiteWeight can be accurate at any lag time, the range of lag times could be explored to optimize numerical efficiency in equilibrium calculations.  

We foresee diverse opportunities for applying RiteWeight in the future. 
To promote the longstanding goal of generating Boltzmann-weighted ensembles for proteins, 
RiteWeight could be applied to MD trajectories initiated with machine-learning approaches based on AlphaFold \cite{ai_feig2023ensemble,ai_cortes2025alphafold_ensembles,ai_jaakkola024alphafold_ensembles,ai_meiler2023protein_states_ensembles} or initiated with nuclear magnetic resonance data \cite{nmr_angyan2013ensemble,nmr_kannan2014conformational}.
Further, RiteWeight could be used to iteratively calculate local mean first passage times and generate transition paths in combination with the weighted ensemble method \cite{aristoff2023recent_math,ryu2025reducing}.

In conclusion, the new RiteWeight algorithm largely solves the problem of reweighting flawed initial data sets into desired steady state distributions.  The iterative use of random clusters leads to quasi-continuous and accurate final distributions.  Applications to equilibrium and nonequilibrium observables compare favorably to state-of-the-art methods, with a particular advantage in quantifying nonequilibrium net fluxes.

\section*{Data Availability Statement}
The python package for running the RiteWeight algorithm is available at \url{https://github.com/ZuckermanLab/rite_weight}. The data used in this study to evaluate RiteWeight’s performance is available at \change{\url{https://doi.org/10.5281/zenodo.18489245}}. The SynMD Trp-cage model used in this study is hosted at \url{https://github.com/jdrusso/SynD-Examples}

\section*{Acknowledgments}
We thank Jeremy Copperman and John Russo for valuable discussions on the subject of reweighting trajectories. We thank DE Shaw Research for the Trp-cage trajectory data. The authors are grateful for financial support from the NIH under Grant GM115805. \change{The research reported in this publication used computational infrastructure supported by the Office of Research Infrastructure Programs, Office of the Director, of the National Institutes of Health under Award Number S10OD034224.}

\bibliography{main}
\bibliographystyle{unsrt}


\section*{Appendix: Analysis of RiteWeight fixed point}

This section identifies the fixed point of the RiteWeight algorithm, which is defined as follows.
\begin{definition}[Fixed point]
A sequence of trajectory weights $(w_i)_{i=1}^N$ with $\sum_i w_i = 1$ is a fixed point of RiteWeight if the weights stay the same during each RiteWeight iteration for each positive probability choice of partition.
\end{definition}

In practice, we typically apply RiteWeight in a continuous state space with clusters defined by a random Voronoi tesselation.
Here for simplicity our analysis considers a finite state space with clusters defined by a random hyperplane tesselation.
See \cite{oreilly2022stochastic} for background on random hyperplane tesselations, which are called ``stable under iteration tesselations'' in stochastic geometry.

\begin{assumption}[Finite state space]
\label{assume:finite}
The state space consists of a finite number of distinct microstates $\alpha \in \mathbb{R}^d$.
\end{assumption}

\begin{assumption}[Random hyperplane tesselation]
\label{assume:random}
The RiteWeight clusters are generated through one or more iterations of the following procedure.
Initially, there is a single cluster containing all the microstates. 
At each iteration, any cluster $C$ that contains at least two microstates is randomly split by a hyperplane into two new clusters as follows. 
First, we choose a uniformly random direction
\begin{equation*}
    u \sim \operatorname{Unif}\{v \in \mathbb{R}^d \,:\,\lVert v \rVert = 1\}.
\end{equation*}
Then, conditional on $u$, we choose a uniformly random offset 
\begin{equation*}
    \gamma \sim \operatorname{Unif}\Bigl\{\eta \in \mathbb{R} \,:\, \min_{\alpha \in C} \alpha^\top u < \eta <  \max_{\alpha \in C} \alpha^\top u\biggr\}.
\end{equation*}
The normal direction and the offset define two new clusters of microstates, 
\begin{equation*}
    \{\alpha \in C:\alpha^\top u < \gamma\}
    \quad \text{and} \quad \{\alpha \in C:\alpha^\top u > \gamma\}
\end{equation*}
that are split by the hyperplane. 
\end{assumption}

Under Assumptions \ref{assume:finite} and \ref{assume:random},
the following main result shows that the RiteWeight fixed point is determined by the microstate transition matrix.
\change{The microstate transition matrix does not change with the RiteWeight iterations.
Therefore, assuming RiteWeight converges, it can only converge to a single fixed point.}   

\begin{theorem}[RiteWeight fixed point] \label{thm:riteweight}
    Assume a finite state space (Assumption~\ref{assume:finite}) and a random hyperplane model (Assumption \ref{assume:random}).
    Consider a sequence of trajectories with weights $(w_i)_{i=1}^N$
    satisfying $\sum_i w_i = 1$, and define the microstate transition matrix $\boldsymbol{P}$ with entries
    \begin{equation} \label{eq:defP}
        P_{\alpha \beta} = \frac{\sum_{i_1 \in \alpha, \, i_2 \in \beta} w_i}{\sum_{i_1 \in \alpha} w_i}.
    \end{equation}
    Assume $\boldsymbol{P}$ has a unique stationary measure.
    Then the sequence of weights $(w_i)_{i=1}^N$ is the fixed point of RiteWeight if and only if the vector $\boldsymbol{\mu}$ with entries
    \begin{equation}
       \mu_\alpha = \sum_{i_1 \in \alpha} w_i
    \end{equation}
    is the stationary measure of $\boldsymbol{P}$, that is,
    \begin{equation}
    \label{eq:stationary_weights}
        \boldsymbol{\mu}^\top \boldsymbol{P} = \boldsymbol{\mu}^\top.
    \end{equation}
\end{theorem}
\begin{proof}
First we check that equation~\eqref{eq:stationary_weights} implies $(w_i)_{i=1}^n$ is the fixed point of RiteWeight.
This happens if the weights $w_I = \sum_{i_1 \in I} w_i$ appear as entries of the left leading eigenvector of the cluster transition matrix $\boldsymbol{T}$.
By assumption, the microstate transition matrix $\boldsymbol{P}$ has a unique stationary measure, so it \change{has one closed communicating class.}
It follows that the cluster transition matrix $\boldsymbol{T}$ \change{has one closed communicating class,} so $\boldsymbol{T}$ has a unique stationary measure also.
Next, equation~\eqref{eq:stationary_weights} implies
\begin{equation*}
\label{eq:stationarity}
    \sum_I w_I T_{IJ}
    = \sum_{\alpha} \sum_{\beta \in J}\mu_\alpha P_{\alpha \beta}
    = \sum_{\beta \in J} \mu_\beta
    = w_J,
    \quad \text{for each cluster } J.
\end{equation*}
This confirms that the left leading eigenvector has the correct entries, so $(w_i)_{i=1}^n$ is the fixed point of RiteWeight.

Next, we assume that $(w_i)_{i=1}^n$ is the fixed point of RiteWeight and check that equation~\eqref{eq:stationary_weights} holds.
Since $(w_i)_{i=1}^n$ is the fixed point, each positive-probability cluster $J$ must satisfy
\begin{equation*}
    \sum_{\alpha \in J} \mu_\alpha = w_J = \sum_I w_I T_{IJ} 
    = \sum_\alpha \sum_{\beta \in J} \mu_\alpha P_{\alpha\beta}.
\end{equation*}
By considering all positive probability clusters, we arrive at the identity
\begin{equation*}
    \boldsymbol{\mu}^\top \boldsymbol{A} = \boldsymbol{\mu}^\top \boldsymbol{P} \boldsymbol{A}
\end{equation*}
where $\boldsymbol{A}$ is the matrix whose columns are characteristic functions for each positive probability random cluster:
$1$ values indicate membership in the cluster and $0$ values indicate non-membership.
To complete the proof, we will show that $\boldsymbol{A}$ has full column rank and therefore \eqref{eq:stationary_weights} holds. 

We observe that the first iteration of the random hyperplane tesselation generates a uniformly random direction $u \in \mathbb{R}^d$ that leads to distinct values $\alpha^\top u$ for every distinct microstate $\alpha$ with probability one.
Hence, there is an ordering of microstates, say $\alpha_1, \ldots, \alpha_n$, so that the event
\begin{equation*}
    \alpha_1^\top u < \cdots < \alpha_n^\top u,
\end{equation*}
occurs with positive probability.
It follows that a division into clusters
\begin{equation*}
    \{\alpha_1, \ldots, \alpha_i\}
    \quad \text{and} \quad
    \{\alpha_{i+1}, \ldots, \alpha_n\}
\end{equation*}
occurs with positive probability for each $i=1,\ldots,n-1$. 
With the microstates ordered in this way, the linear span of the columns of $\boldsymbol{A}$ includes all the vectors $\sum_{j=1}^i \boldsymbol{e}_j$, where $\boldsymbol{e}_i$ is a standard basis vector.
This shows
$\boldsymbol{A}$ has full column rank and equation~\eqref{eq:stationary_weights} holds, completing the proof.
\end{proof}

\section*{Supplementary figures}

\setcounter{figure}{0}
\renewcommand{\thefigure}{S\arabic{figure}} 
\renewcommand{\figurename}{Figure}

\begin{figure}[h]
\centering
\includegraphics[width=\linewidth]{figures/KL_Divergence_Cluster1000_C10.png}
\caption{Convergence of the SynMD Trp-cage equilibrium distribution estimated using RiteWeight.
The vertical axes show the symmetric Kullback-Leibler (KL) divergence between the estimated stationary distribution in the current iteration vs. the estimated distribution 100 iterations prior.
The horizontal axes show the number of iterations. (a) The KL divergence using 1000 clusters. (b) The KL divergence using 10 clusters.}
\label{fig:synmd-converge-kl}
\end{figure}
 
\begin{figure}[h]
\centering
\includegraphics[width=\linewidth]{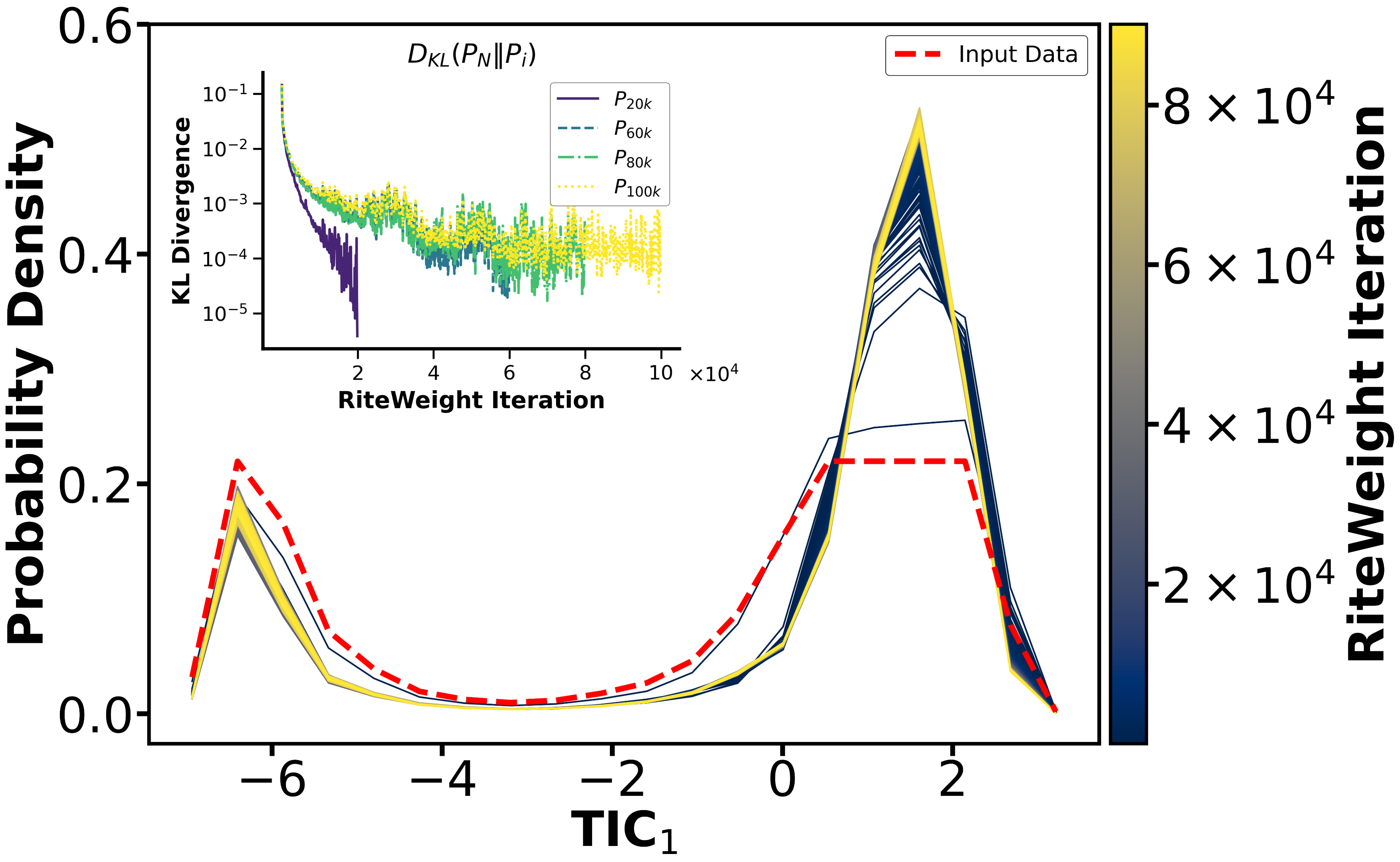}
\caption{Convergence of the atomistic Trp-cage equilibrium distribution using RiteWeight. The main plot shows the convergence of the initial (dashed red) distribution to the final (yellow) distribution using increments of 100 iterations.  The inset shows the KL divergence of the RiteWeight distribution with reference to iteration 20,000, 60,000, 80,000 or 100,000, as shown in the legend.  The similarity of the curves for 60,000 iterations and beyond suggests RiteWeight is converged.
}
\label{fig:RiteWeight_Eq_Conv}
\end{figure}

\begin{figure}[h]
\centering
\includegraphics[width=1.0\linewidth, trim=50 30 50 40,clip]{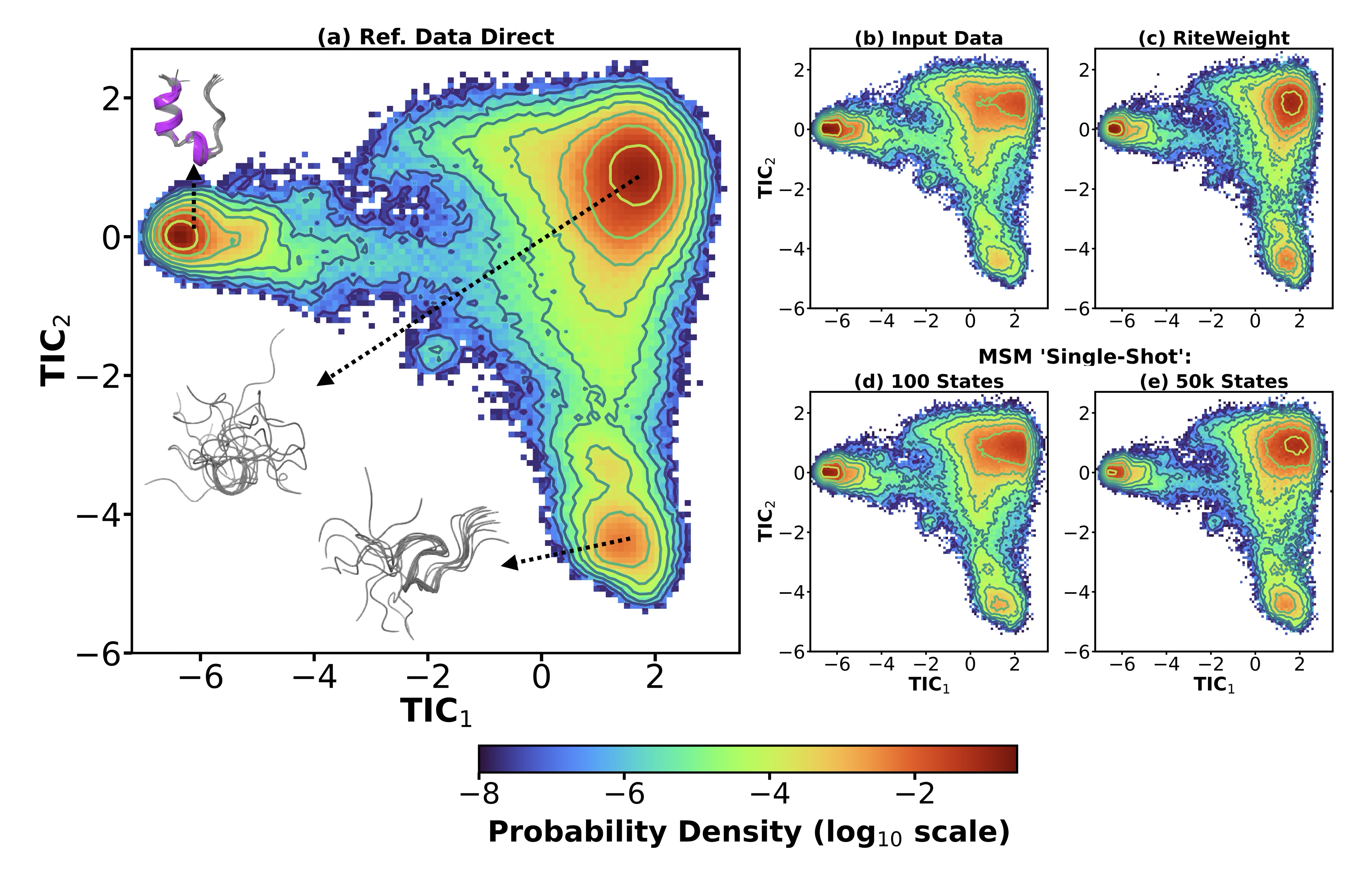}
\caption{\change{Equilibrium probability densities projected onto the two slowest tICA (TIC) components. Contour lines indicate $0.82$-unit increments in the base-10 logarithm of the probability density. (a) The reference density is calculated from a 208 $\mu$s MD trajectory. Ten randomly selected configurations are overlaid for each of the most probable states.
(b) The input data is mis-distributed relative to the reference. (c) RiteWeight successfully recovers the true equilibrium distribution. (d–e) MSM `single-shot' estimates using coarse (100 states) and fine (50,000 states) resolutions fail to recover features of the reference  distribution near the energy basins. The same data employed for panels (b–e) is shown projected onto TIC$_1$ in Fig.~\ref{fig:equil-MD} of the main text.}}
\label{fig:2D_landscape}
\end{figure}

\begin{figure}[t]
\centering
\includegraphics[width=0.7\linewidth]{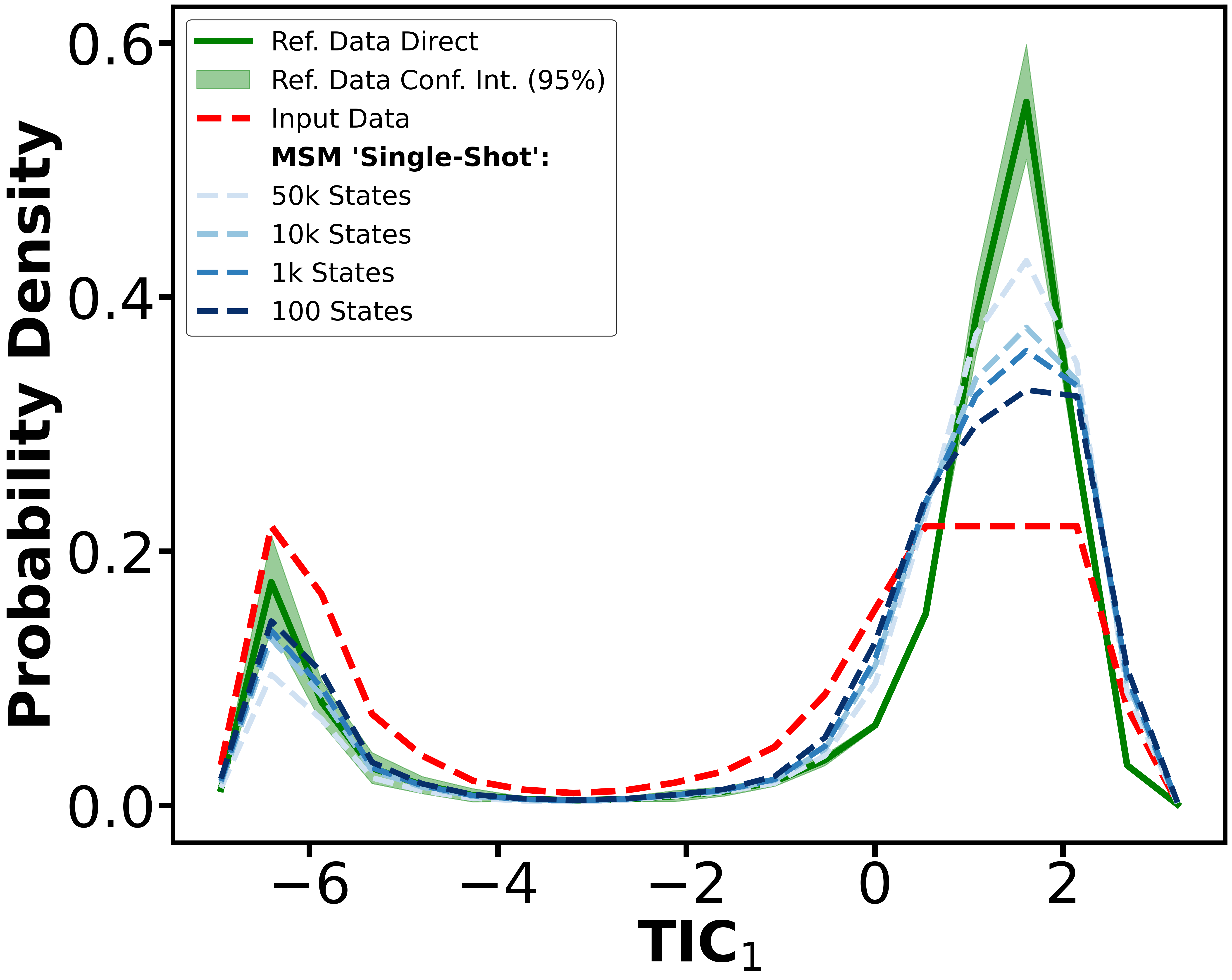}
\caption{Markov state model ``single shot'' estimates for the equilibrium stationary distribution based on different numbers of \change{MSM states} (blue dashed lines) using the lag time $\tau$ = 100 ns. 
}
\label{fig:Eq_MSM_TR_100ns}
\end{figure}

\begin{figure}[h]
\centering
\includegraphics[width=0.8\linewidth,trim=50 30 50 52,clip]{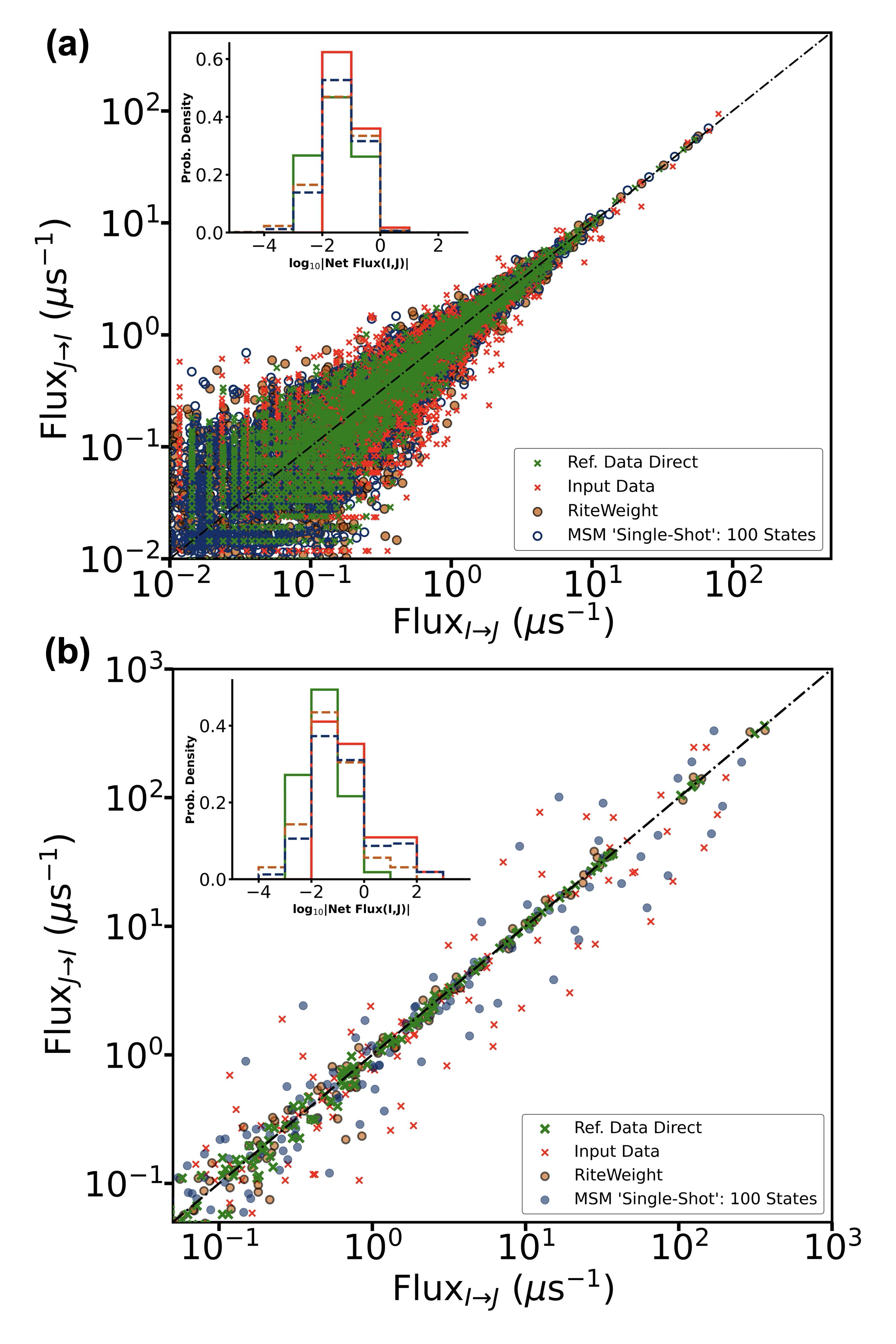}
\end{figure}
\FloatBarrier
\captionsetup{type=figure}
\captionof{figure}{\change{Detailed-balance assessment at fine and coarse scales. 
Here, $\mathrm{Flux}_{I\rightarrow J}$ is the total weight of trajectory segments transitioning from state $I$ to $J$ over a lag time of $\tau$ = 10 ns.
The violation of detailed balance is measured by the distance of data points from the diagonal.
The insets show histograms of the net flux magnitudes $|\text{Net Flux}(I,J)| = |\text{Flux}_{I\rightarrow J} - \text{Flux}_{J \rightarrow I}|$
for the reference data (green solid line), mis-distributed input data (red solid line), RiteWeight estimates (dark orange dashed line), and MSM ``single-shot'' estimates based on 100 states (blue dashed line).
(a) Fine-scale regions of configuration space are defined by the 100 MSM states used to discretize the equilibrium distribution in Fig.~\ref{fig:equil-MD} of the main text.
At the fine scale, MSMs and RiteWeight both shift to smaller violations compared to the input data, even though MSMs explicitly enforce detailed balance during stationary-distribution estimation \cite{noe2015pyemma} whereas RiteWeight does not.
(b)
Coarse-scale regions are defined by 20 uniform bins along TIC$_1$. 
The fluxes estimated by RiteWeight are similar to those from the reference data and lie close to the diagonal; hence RiteWeight largely restores detailed balance in this coarser analysis, whereas MSM estimates fail to recover detailed balance.
}}
\label{fig:detail_balance}

\begin{figure}[t]
\centering
\includegraphics[width=0.7\linewidth]{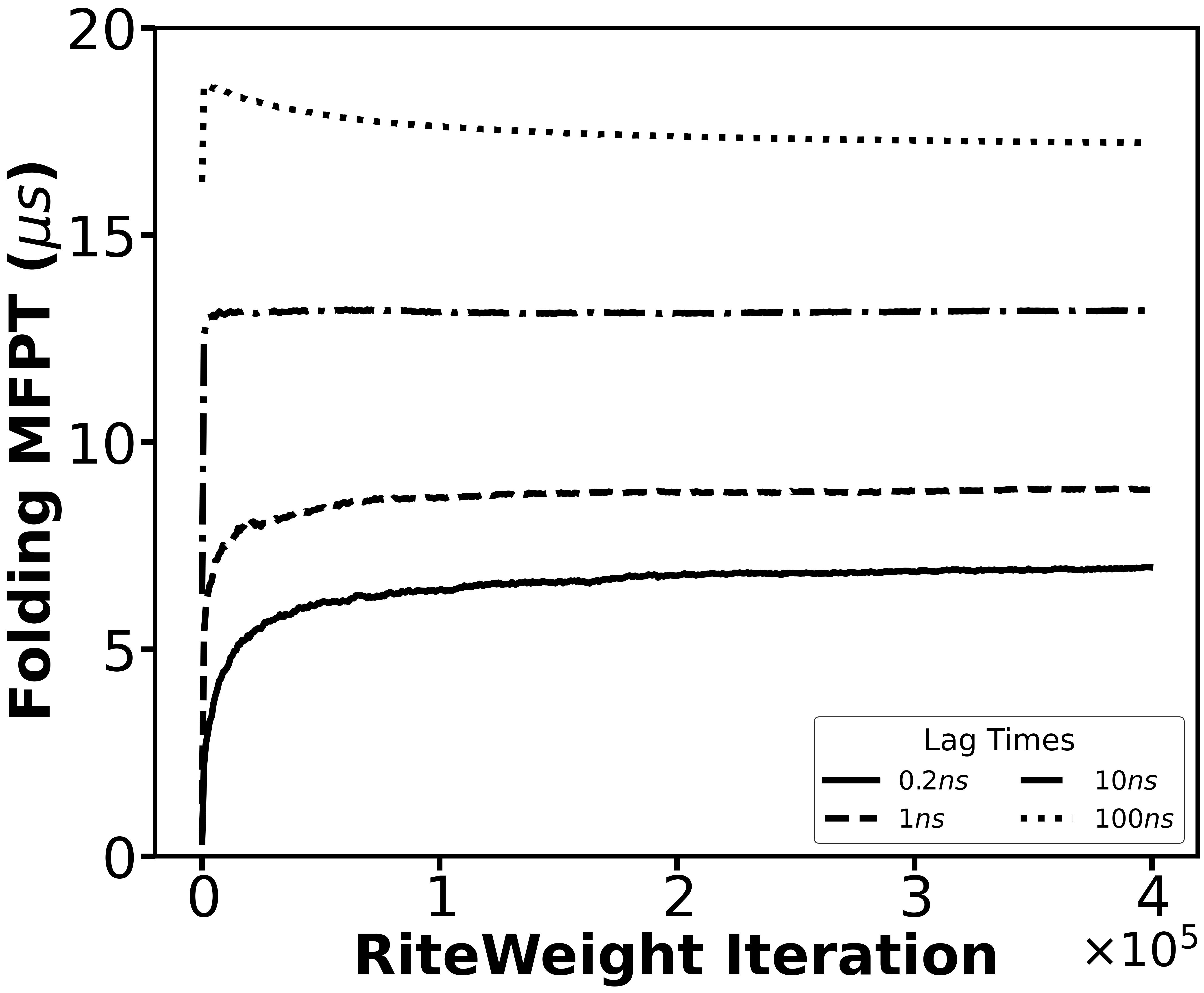}
\caption{Convergence of the atomistic Trp-cage folding MFPT using RiteWeight.
Four different lag times are examined, and the MFPT is steady after 200,000 iterations for any lag time. 
The variation of MFPT with lag time is not an artifact but a physical consequence of the fact that first folding events sometimes are missed with longer lag times.
}
\label{fig:MFPT_Conv}
\end{figure}

\begin{figure}[h]
\centering
\includegraphics[width=1.0\linewidth,trim=28 30 50 40,clip]{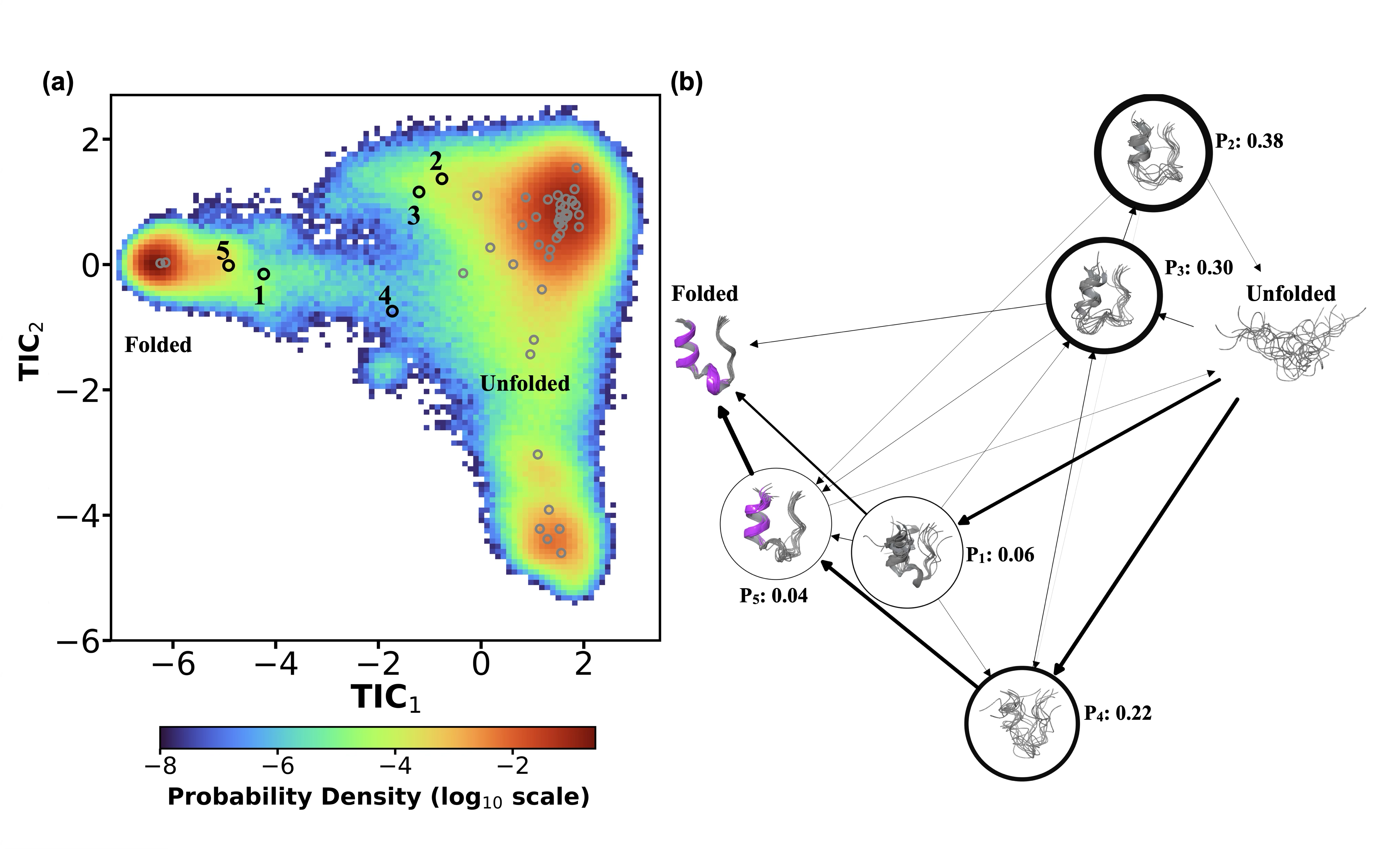}
\caption{\change{Folding intermediates and sequences of transitions constructed from reference data at a $0.2$ ns lag time.
(a) The equilibrium probability density is projected onto the two slowest tICA components. 
Black circles indicate MSM states selected as intermediate macrostates, while gray circles denote MSM states assigned to the unfolded and folded macrostates.
(b) In the nonequilibrium flux network, arrow thicknesses indicate net fluxes.
Node thicknesses indicate non-equilibrium steady state probabilities for the intermediate macrostates, which are normalized to sum to one for clarity.
These non-equilibrium probabilities are distinct from the equilibrium probability density in panel (a).
Ten randomly selected configurations are overlaid for each macrostate, with purple regions highlighting folded helical segments.
Cluster $2$ is partially unfolded;
clusters $1$, $3$, and $5$ feature a partially folded N-terminus; and cluster $4$ identifies the formation of the hydrophobic core.}}
\label{fig:flux_network}
\end{figure}

\begin{figure}[h]
\centering
\includegraphics[width=0.8\linewidth]{figures/net_flux_with_CI_lines_TR0_2ns.png}
\includegraphics[width=0.8\linewidth]{figures/net_flux_with_CI_lines_TR1ns.png}
\caption{Net flux analysis showing confidence intervals for MD-based estimates.
For each pair of \change{macro}states, the net flux is shown for RiteWeight (dark orange), Markov state models (MSM -- blue), and molecular dynamics (MD -- X and confidence interval).
The MD estimates are calculated from the ensemble of reactive trajectories that proceed from the unfolded to the folded macrostate.
The MD 95\% confidence interval is derived by bootstrapping the 24 round-trip paths starting from the folded state, entering the unfolded state, and then ending back at the folded state.
Two lag times are used for analysis: 0.2 ns (top) and 1.0 ns (bottom).}
\label{fig:net_flux_CI_short_lags}
\end{figure}

\begin{figure}[h]
\centering
\includegraphics[width=0.8\linewidth]{figures/net_flux_with_CI_lines_TR10ns.png}
\includegraphics[width=0.8\linewidth]{figures/net_flux_with_CI_lines_TR100ns.png}
\caption{Net flux analysis showing confidence intervals for MD-based estimates.
Two lag times are used for analysis: 10 ns (top) and 100 ns (bottom).
}
\label{fig:net_flux_CI_long_lags}
\end{figure}

\end{document}